\documentclass[prl,twocolumn, aps,preprintnumbers,showpacs, superscriptaddress, 10pt]{revtex4-1}

\usepackage{latexsym}
\usepackage{amsmath}
\usepackage{amssymb}
\usepackage{amsfonts}
\usepackage{ifthen}
\usepackage{revsymb}
\usepackage{yfonts}
\usepackage{graphicx}
\usepackage[all]{xy}

\usepackage{amsmath}
\usepackage{amssymb}
\usepackage{amsthm}

\newcommand{\be}{\begin{equation}}
\newcommand{\ee}{\end{equation}}
\newcommand{\ba}{\begin{array}}
\newcommand{\ea}{\end{array}}
\newcommand{\bea}{\begin{eqnarray}}
\newcommand{\eea}{\end{eqnarray}}

\newcommand{\calB}{{\cal B }}

\newcommand{\calG}{{\cal G }}

\newcommand{\ZZ}{\mathbb{Z}}

\newcommand{\nn}{\nonumber}

\newtheorem{dfn}{Definition}
\newtheorem{lemma}{Lemma}
\newtheorem{prop}{Proposition}
\newtheorem{theorem}{Theorem}

\begin{document}

\title{On the energy landscape of 3D spin Hamiltonians with topological order}

\author{Sergey \surname{Bravyi}}
\affiliation{IBM Watson Research
Center,  Yorktown Heights,  NY 10598}
\author{Jeongwan \surname{Haah}}
\affiliation{Institute for Quantum Information, California Institute of Technology, Pasadena, CA 91125}

\date{20 May 2011}

\begin{abstract}

We explore feasibility of a quantum self-correcting memory based on 3D spin Hamiltonians
with topological quantum order in which thermal diffusion of topological defects is suppressed
by macroscopic energy barriers. To this end we  characterize the energy landscape of  stabilizer code Hamiltonians
with local bounded-strength interactions which have a
topologically ordered ground state but do not have string-like logical operators. We prove that any sequence of local errors mapping a ground state of such Hamiltonian to an orthogonal ground state must cross an energy barrier growing at least as a logarithm of the lattice size.
Our bound on the energy barrier is shown to be tight up to a constant factor for one particular 3D spin Hamiltonian.
\end{abstract}


\maketitle
Topologically ordered phases of matter  display a variety of fascinating properties having no counterpart in the classical
physics. Most notable ones are topological invariants
such as the the Hall conductance~\cite{Girvin:1986}, ground state degeneracy~\cite{Einarsson:1990},
and topological entanglement entropy~\cite{LevinWen:2006,KitaevPreskill:2006}
which are insensitive to generic local perturbations~\cite{Wen:1990,Kitaev:1997,Bravyi:2010}.
The intrinsic stability against perturbations motivated interest to topological phases
as a storage medium for a reliable quantum memory~\cite{Kitaev:1997,Dennis:2001,jiang:2008} and as a platform for
 building a topological quantum computer~\cite{Kitaev:1997,Freedman:2004}.

A big open question in the theory of topological quantum order (TQO) concerns feasibility of a non-volatile, or, {\em self-correcting}, quantum memory~\cite{Kitaev:1997,Bacon:2005}. Such a memory would permit reliable long-term storage of quantum information in a presence of sufficiently weak thermal noise without need for active stabilization and error correction during the storage period. The main challenge in designing Hamiltonians with self-correcting properties
is to combine TQO with an energy landscape that could prevent errors caused by thermal fluctuations
from accumulating. This could guarantee that the error density remains sufficiently small during the entire
storage period and the encoded information can be safely extracted from the memory by performing
an active error correction at the read-out phase.

In spite of being intrinsically stable against perturbations at the zero-temperature,  TQO models display extreme fragility against thermal fluctuations~\cite{Nussinov:2008} suggesting impossibility of quantum self-correction.
A thermal stability analysis involving finite-temperature extensions of the topological entanglement entropy
has been undertaken for the 2D and 3D toric code models by Castelnovo and Chamon~\cite{Castelnovo:2007},
and by Iblisdir et al~\cite{Iblisdir:2010}. These models were shown to undergo
a  transition from  a topologically ordered phase at $T=0$ to a different phase with either partial or no
topological order at  any positive temperature~\cite{Castelnovo:2007,Castelnovo:2008,Iblisdir:2010}.

The first rigorous analysis of self-correcting properties for the toric code models was performed by  Alicki et al~\cite{AFH:2009,Alicki:4D}. It showed that the 4D toric code Hamiltonian has self-correcting properties
for sufficiently small temperature, while 2D and 3D toric codes are not self-correcting at any
finite temperature.
The ideas of~\cite{AFH:2009,Alicki:4D} were developed further by Kay~\cite{kay:nonreliable}, Chesi et al~\cite{Chesi:2010,Chesi:tc}
and Pastawski et al~\cite{Pastawski:2009}.

The main feature of the 4D toric code model responsible for self-correction is a {\em macroscopic energy
barrier} that must be crossed by any sequence of local errors whose combined action on encoded states
cannot be corrected at the final read-out phase~\cite{Alicki:4D}. The height of this barrier grows
linearly with the lattice size due to a finite string-tension characterizing boundaries of membranes
associated with errors.  It  is analogous to the energy barrier
separating ground states with positive and negative magnetization in the ferromagnetic 2D Ising model.
Unfortunately, this behavior cannot be reproduced in any known 2D or 3D model due to a presence
of point-like excitations carrying a non-trivial topological charge, or, {\em point-like defects}~\footnote{One should not confuse point-like defects in TQO models
whose creation requires constant energy
and  vortex-type topological excitations in classical spin systems
whose creation requires divergent energy for a single isolated vertex.}.
These defects  are analogous to domain-walls in the 1D Ising model --- a single isolated defect
has only a  constant energy cost, but its creation requires a highly non-local operation affecting
a macroscopic number of qubits (spins).
Whether or not the presence of point-like defects rules out self-correcting properties
may depend on how fast these defects can diffuse across the system.
For example, Hamma et al~\cite{Hamma:toric-boson} used a coupling with a bosonic field
to create an effective long-range attractive interaction between defects whereby suppressing the diffusion.
A different possibility is realized in the 3D  Chamon's model~\cite{Chamon:2005,BLT:2011}.
This model offers a topological protection against diffusion of some types of defects (but not all of them).
These defects, called monopoles in~\cite{BLT:2011}, can be created at corners of rectangular
shaped membranes. A hopping of a single isolated monopole between adjacent lattice sites
is a highly non-local operation  affecting a macroscopic number of qubits, see~\cite{BLT:2011} for details.

In the present paper we propose yet another possibility to suppress the diffusion of defects that can be
realized in certain 3D spin Hamiltonians with strictly local bounded-strength interactions.
These Hamiltonians, associated with stabilizer error correcting codes~\cite{Gottesman:1998},
have a peculiar property that isolated defects
cannot move further than a certain constant distance away without creating other defects.
For brevity, we shall refer to this property (stated more formally below) as a {\em no-strings rule}
because it is closely connected to the absence of logical string-like operators capable of
moving the defects. Let us point out that the first example of a 3D spin Hamiltonian
with TQO obeying the no-strings rule has been found only quite recently by one of us~\cite{Haah:2011}.
We prove that any sequence of local errors creating an isolated defect from the vacuum
with no other defects within distance $R$ must cross an energy barrier at least  $c\log{R}$
for some constant $c$.
Furthermore, the length of such error sequence must grow at least as $R^\gamma$ for some constant $\gamma > 1$.
The same bound applies to creation of any isolated cluster of defects with a non-trivial total
topological charge.

It shows that although defects do not interact directly, their diffusive motion is suppressed by logarithmic energy barriers preventing the defects from spreading
 (the concept of a diffusive motion must be used with care
in our case because individual defects can only move a constant distance away).

We also prove a logarithmic lower bound on the energy barrier
for implementing any logical operator. More precisely, we prove
 that any sequence of local errors mapping a ground state
to an orthogonal ground state must cross an energy barrier at least $c\log{L}$, where
$L$ is the lattice size and $c$ is some constant. For the Hamiltonian discovered
in~\cite{Haah:2011} this bound is tight up to a constant factor.
 Although the scaling of the energy barrier is not as favorable
as the one in the  4D toric code, we point out that
the energy barrier does not grow with the lattice size at all
for all previously studied TQO Hamiltonians in the 2D and 3D geometry.
A naive estimate of the storage time $\tau$ for a memory with an energy barrier $B$
operating at a temperature $T$ can be
made using the Arrhenius law, namely, $\tau \sim e^{B/T}$. Since in our case $B=c\log{L}$ for
some constant $c$, we arrive at $\tau\sim L^{c/T}$.
Although this `derivation' gives only polynomial scaling of $\tau$,
the degree of the polynomial can be made arbitrarily large
by choosing sufficiently small temperature.

It is worth pointing out that 2D Hamiltonians with TQO
always have string-like logical operators and thus 3D is the smallest spatial dimension
for constructing Hamiltonians obeying the no-strings rule.
Indeed, it was shown by Terhal and one of us~\cite{BT:2009} that for
any 2D local stabilizer-type Hamiltonians the energy barrier for implementing at least
one logical operator is constant.
It should also be noted that 3D translation-invariant stabilizer Hamiltonians with TQO
can obey the no-strings rule only if the ground state degeneracy is not invariant under changing lattice dimensions
~\cite{Yoshida:2011,Yoshida:3Dnew}.

An obvious difficulty in proving a lower bound on the energy barrier is that there are too many ways
to choose an error path (i.e., a sequence of local errors) implementing a fixed Pauli operator.
In fact, we do not impose any restrictions on the length of the error path except that it is finite.
Our bound should be applicable to any such error path.
An additional difficulty is that a Pauli operator that an error path needs to implement is only defined
modulo stabilizers.  We resolve these difficulties using a novel technique
which can be regarded as a renormalization group in the space of error paths,
see Fig.~\ref{fig:RG} for a brief summary of the technique.
Let us now state our main results more formally.

\begin{figure}[h]
\centerline{\includegraphics[height=4cm]{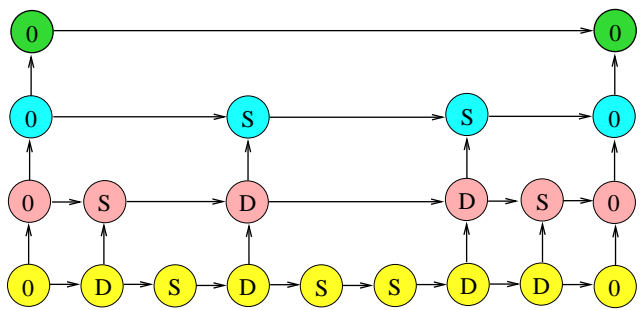}}
\caption{Renormalization group technique used to prove a
logarithmic  lower bound on
the energy barrier. Horizontal axis represents time. Vertical axis
represents RG level $p=0,1,\ldots,p_{max}$.
A sequence of level-$0$ errors (single-qubit Pauli operators) implementing a logical operator $\overline{P}$
defines a level-$0$ syndrome history (yellow circles)  that consists of sparse (S) and dense (D) syndromes.
The history begins and ends with the vacuum ($0$).
For any level $p\ge 1$ we define a level-$p$ syndrome history
by retaining only dense syndromes at the lower level.
A syndrome is called dense at level $p$ if it
cannot be partitioned into clusters of size $\le (10\alpha)^p$
separated by distance $\ge (10\alpha)^{p+1}$,
where $\alpha$ is a constant coefficient from the no-strings rule.
Each  level-$p$ error (horizontal arrows) connecting syndromes $S',S''$
is equivalent to the product of all level-$(p-1)$ errors between $S',S''$
modulo a stabilizer. We prove that these stabilizers can be
chosen such that
level-$p$ errors act on $2^{O(p)}$ qubits.
Since at the highest  level $p=p_{max}$
a single level-$p$ error is a logical operator,
one must have $p_{max}=\Omega(\log{L})$.
We prove that level-$p$ dense syndromes contain $\Omega(p)$ defects
which implies that at least one syndrome at level $p=p_{max}-2$
consists of $\Omega(\log{L})$ defects.}
\label{fig:RG}
\end{figure}

{\bf Stabilizer code Hamiltonians.}
Quantum spin Hamiltonians based on stabilizer codes provide convenient mathematical models of TQO that became an active research topic over last years.
We consider a regular $D$-dimensional cubic lattice $\Lambda$
with periodic boundary conditions and linear size $L$, that is, $\Lambda=\ZZ_L^D$.
Each site $u\in \Lambda$
is populated by a finite number of qubits.
A large class of TQO models can be described by the so-called
{\em stabilizer Hamiltonians} defined as
\be
\label{H}
H=-\sum_{a=1}^{M} G_a,
\ee
where each term $G_a$  is a multi-qubit Pauli operator (a tensor product of $I,X,Y,Z$
with an overall $\pm 1$ sign)
and different terms commute with each other. Hence
\be
G_a G_b =G_b G_a  \quad \mbox{and} \quad G_a^2=I.
\ee
The abelian group $\calG$ generated by $G_1,\ldots,G_M$
is called a {\em stabilizer group} of the code. Elements of $\calG$
are called {\em stabilizers}.
Physically realistic Hamiltonians  may only involve interactions between
small subsets of qubits located close to each other. We shall
assume that each generator $G_a$
acts non-trivially (by $X,Y$ or $Z$) only on a set of qubits located at vertices of
an elementary cube. It is allowed to have more than one generator per cube.
Any short-range stabilizer Hamiltonian can be written in this form by performing
a coarse-graining of the lattice. The Hamiltonian may or may not be translation-invariant.

We shall assume that $H$ is frustration-free~\footnote{This is always the case
for independent generators $G_a$. Since our goal is to obtain a lower bound on the
energy barrier, we can assume that the generators are independent, although it does not
play any role in our analysis.}, that is,
ground states $\psi_0$ of $H$ obey $G_a\, \psi_0= \psi_0$
for all $a$, or, equivalently, the stabilizer group $\calG$
does not contain $-I$.
Consider any
multi-qubit Pauli operator $E$.
A state $\psi = E\,  \psi_0 $
is an excited eigenstate of $H$. Obviously, $G_a \, \psi=\pm \psi$
where the sign depends on whether $G_a$ commutes (plus)
or anticommutes (minus) with $E$.
 Any flipped generator ($G_a\, \psi=-\psi$) will be referred to as a {\em defect}.
Let us emphasize  that a configuration of  defects in $\psi$ is the same for all
ground states $\psi_0$.
An eigenstate with $m$ defects has energy $2m$ above the ground state.
For brevity, we shall use the term {\em vacuum} for a  ground state of $H$
whenever its choice  is not important.
A Pauli operator $E$ whose action on the vacuum creates no defects is either a stabilizer ($E\in \calG$),
or a {\em logical operator} ($E\notin \calG$, but $E$ commutes with $\calG$).
In the former case any ground state of $H$ is invariant under $E$.
In the latter case $E$ maps some ground state of $H$ to an orthogonal ground state.

{\bf Topological order.}
A Hamiltonian is  said to have topological order if it
has a degenerate ground state and different ground states are locally indistinguishable.
We shall need a slightly stronger version of this condition
that involves properties of both ground and excited states.
These properties depend on a length scale $L_{tqo}$ that must
be bounded as $L_{tqo}\ge  L^\beta$ for some constant $\beta>0$.
(For stabilizer codes, $L_{tqo}$ corresponds to the code distance.)
Most of stabilizer code Hamiltonians with topological order satisfy our conditions
with $L_{tqo}\sim L$.
Our first TQO condition concerns ground states:
\begin{center}
\parbox{8cm}{\em If a Pauli operator $E$ creates no defects when applied to the vacuum
and its support can be enclosed by a cube of linear size $L_{tqo}$ then $E$
is a stabilizer,  $E\in \calG$.}
\end{center}
Our second TQO condition concerns excited states.
A cluster of defects $S$ will be called {\em neutral} if it can be created
from the vacuum by a Pauli operator $E$ whose support can be enclosed
by a cube of linear size $L_{tqo}$ without creating any other defects.
Otherwise we say that $S$ is a {\em charged} cluster.
Given a region $A\subseteq \Lambda$ we shall use a notation
$\calB_r(A)$ for the $r$-neighborhood of $A$, that is, a set of all points that have distance at most $r$ from $A$.
Here and below we use $l_{\infty}$-distance on $\ZZ_L^D$.
We shall need the following condition saying that neutral clusters of defects can be
created from the vacuum locally~\footnote{If a lattice has a boundary, charged defects 
might be created locally on the boundary, as it is the case for the planar version of the toric
code. This is the reason why we restrict ourselves to periodic boundary conditions.}.

\begin{center}
\parbox{8cm}{\em Let $S$ be a neutral cluster of defects and $C_{min}(S)$
be the smallest cube that encloses $S$. Then $S$ can be created
from the vacuum by a Pauli operator supported on $\calB_1(C_{min}(S))$.}
\end{center}

{\bf No-strings rule.}
We can now state the property of having no logical string-like operators.
Informally, it says that applying an operator with a `string-like' support to the vacuum cannot create
charged  defects at the end-points of the string, assuming that the string
is sufficiently long. Let us now define this property rigorously.
Let $E$ be any Pauli operator whose support can be enclosed by a cube
of linear size $L_{tqo}$ and $S$ be a cluster of defects obtained by applying
$E$ to the vacuum. Let $A_1,A_2$ be any pair of disjoint cubes
of the same linear size $\rho$.
We shall say that $E$ is a {\em logical string segment}
with {\em anchor regions} $A_1,A_2$  iff $S$ is contained in the union
$A_1\cup A_2$. Equivalently, $E$ must commute with all generators
$G_a$ located outside $A_1\cup A_2$.
We will say that a logical string segment $E$ has {\em aspect ratio}
$\alpha$ iff the distance between $A_1$ and $A_2$ is at least $\alpha \rho$.
We shall only consider string segments with sufficiently large aspect ratio,
say, $\alpha\ge 1$.
A  logical string segment $E$ is called {\em trivial} iff
the cluster of defects contained inside any anchor region is neutral.
\begin{center}
\parbox{8cm}{\em  A code has no logical string-like operators iff
there exists a constant $\alpha$ such that all
logical string segments with aspect ratio greater than $\alpha$
are trivial.}
\end{center}
In the rest of the paper we shall abbreviate the condition of having no logical string-like operators
as a {\em no-strings rule}.
We note that a 3D stabilizer code (Code~1) discovered in~\cite{Haah:2011}
obeys our topological order conditions with $L_{tqo}\sim L$ and
obeys the no-strings rule with $\alpha=15$.

{\bf Energy barrier.}
Let us consider a process of building  a logical operator $\overline{P}$ from local errors.
It can be described by
an {\em error path} --- a finite sequence of local Pauli errors  $E_1,\ldots,E_T$ such that
$\overline{P}=E_T \cdots E_2 E_1$. For simplicity we shall assume that each local error  $E_t$ is a single-qubit Pauli operator $X$, $Y$, or $Z$. Applying this sequence of errors to a ground state $\psi_0$
generates a sequence of states $\{ \psi(t)\}_{t=0,\ldots,T}$, where
$\psi(0)=\psi_0$ and $\psi(T)=\overline{P}\, \psi_0$ are ground states of $H$, while the
intermediate states $\psi(t)=E_t\cdots E_1 \, \psi_0$ are typically excited.
We say that a logical operator $\overline{P}$ has  energy barrier $\omega$ iff for any error path implementing $\overline{P}$ at least one of the intermediate
states $\psi(t)$ has more than $\omega$ defects.
Note that we do not impose any restrictions on the length of the path $T$ (as long as it is finite).
In particular, one and the same error may be repeated in the error path several times
at different time steps.
We shall also  consider an energy barrier for creating a cluster of  defects
$S$ from the vacuum. We will say that $S$ has energy barrier $\omega$ iff
for any Pauli operator $E$ that creates $S$ from the vacuum and for any
error path implementing $E$ at least one of the intermediate states has more than $\omega$
defects.

Our main results  are the following theorems.
Both theorems apply to any stabilizer Hamiltonian Eq.~\eqref{H}
on a $D$-dimensional lattice that obeys the topological order condition
and the no-strings rule.
\begin{theorem}
\label{thm:main}
The energy barrier for any
logical operator is at least $c \log{L}$,
where $L$ is the lattice size, and $c$ is a constant coefficient.
\end{theorem}
\begin{theorem}
\label{thm:main'}
Let $S$ be a neutral cluster of defects
containing a charged cluster $S' \subseteq S$ of diameter $r$
such that there are no other defects within distance $R$ from $S'$.
If $r+R < L_{tqo}$,
then the energy barrier for creating $S$ from the vacuum is at least $c \log{R}$,
where $c$ is a constant coefficient.
\end{theorem}
\noindent
The constant $c$ depends only on the spatial dimension $D$, the constant
$\alpha$ in the no-strings rule, and the constant $\beta$
in the bound $L_{tqo}\ge L^\beta$.
Since the proof of both theorems uses the same technique, 
below we shall focus on proving Theorem~\ref{thm:main}. Proof of Theorem~\ref{thm:main'}
will require only minor modifications that are explained later on.

\begin{proof}[Proof of Theorem~\ref{thm:main}]
A configuration of defects created by applying a Pauli operator $E$ to the vacuum
will be called a {\em syndrome caused by $E$}.
The process of building up a logical operator $\overline{P}$ by a sequence of local errors
$E_1,\ldots,E_T$
can be described by
a {\em syndrome history}
$\{S(t)\}_{t=0,\ldots,T}$. Here $S(t)$ is the syndrome caused by
the product
$E_t\cdots E_1$, that is,
the partial implementation of $\overline{P}$ up to a step $t$.
The syndrome history starts and ends with the vacuum, i.e.,
$S(0)=S(T)=\emptyset$.
Without loss of generality all intermediate syndromes $S(t)$ are non-empty.
For any integer $p\ge 0$ define  a level-$p$ unit of length as~\footnote{The choice of the constant $10$ 
in the definition of sparsity is somewhat arbitrary. We have not tried
to optimize constants in our proof.}
\[
\xi(p)=(10\alpha)^p, \quad p=0,1,\ldots.
\]
Let $S(t)$ be any non-empty syndrome.
Recall that each defect in $S(t)$ can be associated with some elementary cube of the lattice. 
\begin{dfn}
\label{dfn:psparse}
A syndrome $S(t)$ is called \emph{sparse} at level $p$ iff
the set of elementary cubes occupied by
$S(t)$ can be partitioned into a  disjoint union of clusters
such that each cluster has diameter at most $\xi(p)$
and any pair of distinct clusters combined together has diameter larger
than $\xi(p+1)$. Otherwise $S(t)$ is called dense at level $p$.
\end{dfn}
For example, suppose all defects in $S(t)$ occupy the same elementary cube. 
Since an elementary cube has diameter $1$, such a syndrome $S(t)$ is sparse at any level $p\ge 0$.
If $S(t)$ occupies a pair of adjacent cubes, $S(t)$ is sparse at any level $p\ge 1$,
however it is dense at level $p=0$.
We note that the partition of $S(t)$ into clusters required for level-$p$ sparsity
is unique whenever it exists.
\begin{lemma}
\label{lemma:counting}
Suppose a non-empty syndrome $S(t)$ is dense at all levels $q=0,\ldots,p$.
Then $S(t)$ contains at least $p+2$ defects.
\end{lemma}
\begin{proof}
Let $C^{(0)}_1,\ldots, C^{(0)}_g$ be elementary cubes
occupied by $S(t)$. Obviously, $S(t)$ contains at least $g$ defects. 
Since  $S(t)$ is non-empty and dense at level $0$, we have $g\ge 2$
and there exists a
pair of cubes $C_a^{(0)}, C_b^{(0)}$ such that the union
$C_a^{(0)}\cup C_b^{(0)}$ has diameter at most $\xi(1)$.
Combining the pair $C_a^{(0)}, C_b^{(0)}$  into a single cluster we obtain a partition
$S(t)=C_1^{(1)}\cup \ldots \cup C_{g-1}^{(1)}$
where each cluster $C_a^{(1)}$ has diameter at most $\xi(1)$.
Suppose $S(t)$ is dense at level $1$. Then $g\ge 3$
and there exists
a pair of clusters $C_a^{(1)}, C_b^{(1)}$ such that the union
$C_a^{(1)}\cup C_b^{(1)}$ has diameter at most $\xi(2)$.
Combining the pair $C_a^{(1)}, C_b^{(1)}$ into a single cluster and
proceeding in the same way we arrive at $g\ge p+2$.
\end{proof}
Let us define a level-$p$ syndrome history
as a subsequence of the original syndrome history $\{S(t)\}_{t=0,\ldots,T}$
that includes only those syndromes $S(t)$ that are dense
at all levels $q=0,\ldots,p-1$, see Fig.~\ref{fig:RG}.
The level-$0$ syndrome history includes all syndromes $S(t)$.
The syndrome history starts and ends with the vacuum (empty syndrome)
at any level $p$.
Let $S(t')$ and $S(t'')$ be  a consecutive pair of level-$p$ syndromes.
We define a level-$p$ error $E$ connecting $S(t')$ and $S(t'')$
as the product of all single-qubit errors $E_j$ that occurred between $S(t')$ and $S(t'')$.
Level-$p$ errors are represented by  by horizontal arrows on Fig.~\ref{fig:RG}.
We would like to show that $E$ can be regarded as an approximately local error
on a coarse-grained lattice characterized by the unit of length $\xi(p)$. 
The problem however is that  we do not have any bound on the number
of single-qubit errors $E_j$ in the interval between $S(t')$ and $S(t'')$.
In the worst case, $E$ could act non-trivially on every qubit in the system. 
The following  lemma shows that  level-$p$ errors can be `localized' by multiplying them
with stabilizers.
Let $m$ be the maximum number of defects in the syndrome history, such that  any $S(t)$
contains at most $m$ defects. 
\begin{lemma}
\label{lemma:RG2}
Let $S'\equiv S(t')$ and $S''\equiv S(t'')$ be a consecutive pair of syndromes in the level-$p$ syndrome history.
Let $E$ be the product of all errors $E_j$ that occurred between $S'$ and $S''$.
If $4m (2+\xi(p)) < L_{tqo}$,
then there exists an error $\tilde{E}$  supported on $\calB_{\xi(p)}(S'\cup S'')$
such that $E\tilde{E}$ is a stabilizer.
\end{lemma}
The proof of the lemma, presented in Appendix~A, uses induction in the level $p$.
The proof relies crucially on the no-strings rule.
The latter  asserts that an isolated charged defect belonging to some sparse
syndrome cannot be moved further than distance $\alpha$ away
by a sequence of local errors. Since the no-strings rule is scale invariant,
it can also be applied to a coarse-grained lattice to show that isolated 
charged clusters belonging to some sparse level-$p$ syndrome
cannot be moved further than distance $\alpha \xi(p)$ away, see Appendix~A for details. 

Let $p_{max}$ be the highest RG level, that is, the smallest integer $p$
such that a single level-$p$ error $E$ maps
the vacuum to itself, see Fig.~\ref{fig:RG}.
We claim that $p_{max}=\Omega(\log{L})$.
Indeed, suppose that  $4m(\xi(p_{max})+2) <L_{tqo}$.
Then we can apply Lemma~\ref{lemma:RG2} to the level-$p_{max}$ syndrome history
with $S'=S''=\emptyset$ (vacuum). Lemma~\ref{lemma:RG2} would imply $\tilde{E}=I$, that is,
$E$ must be a stabilizer. On the other hand,
$E$  is equivalent to a logical operator modulo stabilizers. Hence we obtain a contradiction
unless $4m(\xi(p_{max})+2) \ge L_{tqo}$.
We can assume that the maximum number of defects is $m\ll \log{L}$
(if not, there is nothing to prove).
Since $L_{tqo}$ grows as a power of $L$, we conclude that
$p_{max}=\Omega(\log{L})$. The syndrome history
must contain at least one syndrome $S(t)$ which is dense at all
levels $q=0,\ldots,p_{max}-2$ since
otherwise $p_{max}$ could not be the highest RG level.
Lemma~\ref{lemma:counting} then implies that such syndrome $S(t)$
contains $\Omega(\log{L})$ defects proving Theorem~\ref{thm:main}.
\end{proof}

We shall prove Theorem~\ref{thm:main'} using  exactly the same approach as above.
Let $S$ be a neutral cluster of defects and $E$ be a Pauli operator creating $S$ from the vacuum,
with $S' \subset S$ of diameter $r$ being charged.
Consider a hierarchy of syndrome histories similar to the one shown on Fig.~\ref{fig:RG},
where we now maintain the initial syndrome $\emptyset$ and the final syndrome $S$ for all levels.
Let $p_{max}$ be the highest RG level. Then a single  level-$p_{max}$ error $E$
creates $S$ from the vacuum.
Suppose $4m(\xi(p_{max})+2) < L_{tqo}$, where $m$ is the maximum number of defects in the syndrome history.
Lemma~\ref{lemma:RG2} implies that $E$ is the equivalent modulo stabilizers to $\tilde{E}$ supported on $\calB_{\xi(p_{max})}(S)$.
If $\xi(p_{max}) < R/4$, then $\tilde{E}$ must act on two separated regions, one near $S'$ and another far from $S'$.
This means that $S'$ alone can be created by a Pauli operator whose support is enclosed by a cube of linear size $r+R/2$.
Since $r+R < L_{tqo}$, it is contradictory to the assumption that $S'$ is charged.
Therefore, $\xi(p_{max}) \ge R/4$, or $p_{max} = \Omega(\log R)$.
In case where $4m(\xi(p_{max})+2) \ge L_{tqo} > R$, we also have $p_{max} = \Omega(\log R)$ provided $m \ll \log R$.
(There is nothing to prove if $m \sim \log R$.)
Since there must be a syndrome that is dense for all level-$p$ where $p=0,1,\ldots,p_{max}-2$,
Theorem~\ref{thm:main'} follows from Lemma~\ref{lemma:counting}.

{\bf Optimality.}
\newcommand{\drawgenerator}[8]{%
\xymatrixrowsep{.7cm}%
\xymatrixcolsep{.7cm}%
\xymatrix@!0{%
& #8 \ar@{-}[ld]\ar@{.}[dd] \ar@{-}[rr] & & #7 \ar@{-}[ld]  \\%
#1 \ar@{-}[rr] \ar@{-}[dd] &  & #2 \ar@{-}[dd] &            \\%
& #6 \ar@{.}[ld] &  & #5 \ar@{-}[uu] \ar@{.}[ll]       \\%
#3 \ar@{-}[rr] &  & #4 \ar@{-}[ru]                       %
}%
}
\begin{figure}
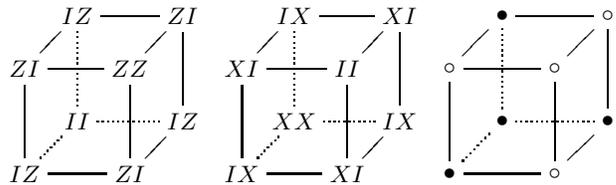

\centering
\begin{tabular}{ccc}
$ \drawgenerator{ZI}{ZZ}{IZ}{ZI}{IZ}{II}{ZI}{IZ} $
&
$ \drawgenerator{XI}{II}{IX}{XI}{IX}{XX}{XI}{IX} $
&
$ \drawgenerator{\circ}{\circ}{\bullet}{\circ}{\bullet}{\bullet}{\circ}{\bullet} $
\end{tabular}
\caption{Stabilizer generator for Code 1 in \cite{Haah:2011}.
This translation-invariant model exhibits TQO, and obeys no-strings rule with $\alpha=15$.
The diagram on the right is a cube in the dual lattice.
The filled dots indicate the defects created by $XI$ at the center of the cube,
which we call level-$0$ pyramid.}
\label{fig:code1_generator_pyramid}
\end{figure}
Our lower bound on the energy barrier is optimal up to a constant factor.
As we noted earlier, Code 1 in \cite{Haah:2011} exhibits TQO
with $L_{tqo}\sim L$  and obeys the no-strings rule with a constant $\alpha=15$.
This code can be described by a stabilizer Hamiltonian Eq.~\eqref{H} with
two qubits per site and two stabilizer generators per cube as shown on Fig.~\ref{fig:code1_generator_pyramid}.
We shall use a notation $PQ_u$ for a two-qubit Pauli operator $P\otimes Q$ applied to the pair of  qubits located at site $u$.
A single-qubit error $XI_u$ creates a cluster of $4$ defects (violated generators of $Z$-type)
located at a set of cubes $S^{(0)}_c=\{c,c+\hat{x},c+\hat{y},c+\hat{z}\}$, where the center of $c$
is obtained from $u$ by a translation $-(\hat{x}+\hat{y}+\hat{z})/2$, see Fig.~\ref{fig:code1_generator_pyramid}.
We shall refer to $S^{(0)}_c$ as a {\em level-$0$ pyramid with an apex} $c$.
For any integer $p\ge 0$ define a {\em level-$p$ pyramid} with an apex $c$
as a cluster of $4$ defects located at a set of cubes  $S^{(p)}_c=\{c,c+2^p\hat{x},c+2^p\hat{y},c+2^p\hat{z}\}$.
We note that a level-$(p+1)$ pyramid with an apex $c$ can be represented as a sum (modulo two)
of four level-$p$ pyramids with an apex at $c$, $c+2^p\hat{x}$,  $c+2^p\hat{y}$, and  $c+2^p\hat{z}$.
Therefore,
$S^{(p)}_c$ can be created from the vacuum by an error $E^{(p)}_u$ defined recursively as
\[
E^{(p+1)}_u = E^{(p)}_u E^{(p)}_{u+2^p\hat{x}} E^{(p)}_{u+2^p\hat{y}} E^{(p)}_{u+2^p\hat{z}}
\]
with $E^{(0)}_u=XI_u$. Simple induction shows that $E^{(p)}_u$ acts on $4^p$ qubits. Its support can be
regarded as a fractal object with a fractal dimension $2$.

Suppose that the lattice has periodic boundary conditions and its linear size is $L=2^n$
for some integer $n$. Then
the four defects of a level-$n$ pyramid cancel each other,
i.e., $S^{(n)}_c=\emptyset$.  It shows that
the error $E^{(n)}_u$ is either a stabilizer or a logical operator.
Consider an auxiliary operator $\overline{Z}_u=\prod_{i,j=1}^{L} ZI_{u-\hat{x}+i\hat{y}+j\hat{z}}$.
One can easily check that $\overline{Z}_u$ commutes with all stabilizer generators
of $X$-type, see Fig.~\ref{fig:code1_generator_pyramid}. On the other hand,
$\overline{Z}_u$ anti-commutes with $E^{(n)}_u$ since their supports overlap
only on a single site $u+(2^n-1)\hat{x}=u-\hat{x}$ and
their restrictions on this site ($ZI$ and $XI$ respectively) anticommute.
It shows that $E^{(n)}_u$ is a logical operator.
Let us now show that $E^{(n)}_u$ has energy barrier at most $4+4\log{L}$.

We claim that the error $E^{(p)}_u$ creating a level-$p$ pyramid $S^{(p)}_c$
has energy barrier $\omega_p=4p+4$. Below we shall only consider
error  paths implementing single-qubit factors of $E^{(p)}_u$ in a certain order.
We shall use induction to prove the following statement:
{\em $E^{(p)}_u$ can be implemented by a sequence of single-qubit errors
without creating more than $\omega_p$ defects. The apex $c$ contains a defect
after each error in this sequence. }
For $p=0$, the claim is obvious.
Suppose we have already proved the claim for $p=0,\ldots,q$.
Let us build a level-$(q+1)$ pyramid $S^{(q+1)}_c$.
We start from building a level-$q$ pyramid $S^{(q)}_c$,
which requires at most $4q+4$ defects and the apex $c$ contains a defect
after each elementary error. Now
we can build another level-$q$ pyramid
$S^{(q)}_{c+2^q\hat{x}}$ using at most $4q+4$ defects
plus the defects that were already present in $S^{(q)}_c$.
Although $S^{(q)}_c$ has $4$ defects, one of them is located at the apex of $S^{(q)}_{c+2^q\hat{x}}$.
The induction hypothesis implies that any partial implementation
of $S^{(q)}_{c+2^q\hat{x}}$ also contains a defect at $c+2^q\hat{x}$. These two defects
cancel each other and must be subtracted from the total cost.
Thus we can build two nearby pyramids at cost of $4+(4q+4)-2=4q+6$ defects.
Both pyramids together have $6$ defects.
The same argument shows that
building the third pyramid $S^{(q)}_{c+2^{q}\hat{y}}$ requires at most $6+(4q+4)-2=4q+8$ defects,
and the three pyramids together have $6$ defects.
Building the last pyramid $S^{(q)}_{c+2^{q}\hat{z}}$
requires at most $6+(4q+4)-2=4q+8$ defects.
Thus, we have constructed a pyramid $S^{(q+1)}_c$ using at most $4(q+1)+4$ defects.
Note that the apex $c$ was occupied by a defect over all stages of this construction.
It proves the step of induction.
We conclude that a logical operator $E^{(n)}_u$ has energy barrier at most
$4n+4 =4+4\log{L}$.

{\bf Separation of charged clusters.}
We have shown that in a process of isolating a charged cluster, there is a logarithmic energy barrier.
The following theorem quantifies how long the process must be.

\begin{theorem}
Let $E$ be a Pauli operator creating $S$, a neutral cluster of defects
containing a charged cluster $S' \subseteq S$ of diameter $r$
such that there are no other defects within distance $R$ from $S'$.
If $r+2R < L_{tqo}$, then the weight of $E$ must be $\ge c R^\gamma$ for some constant $\gamma > 1$ and $c$.
\label{thm:separation_charged_defects}
\end{theorem}

The proof again makes use of renormalization group,
and shows that there is a subset of `fractal dimension' $\gamma > 1$ in the support of $E$.
We assume that $E$ has weight minimum possible.
Let $w$ be a odd positive number.
We say a set of sites $C \subset \Lambda$ is a \emph{level-$p$ chunk} if $\mathrm{diam}(C) < w^p$.
A \emph{path} in the lattice is a finite sequence of sites $(u_1,u_2,\ldots,u_n)$ such that $d(u_i,u_{i+1})=1$.
(Recall that we use the $l_\infty$ metric $d$.)
Equipped with paths, we can say whether a set is connected.

\begin{dfn}
A connected level-$p$ chunk $C \subseteq S$ is \emph{maximal} with respect to a set of sites $S$
if there exist a connected subset $C^\circ \subseteq C$
and a path $\zeta = (u_1,\ldots,m,\ldots,u_n) \subseteq C^\circ$ satisfying
\begin{enumerate}
 \item[(i)]   $d(u_1,u_n)=w^p-w^{p-1}$,
 \item[(ii)]  $d(u_1,m),d(u_n,m)\ge \frac{w^p-w^{p-1}}{2}$,
 \item[(iii)]  $C^\circ$ contains the connected component of $m$ in $\mathcal{B}_{\frac{w^{p}-w^{p-1}}{2}}(m) \cap S$, and
 \item[(iv)] $C$ contains the connected component of $C^\circ$ in $\mathcal{B}_{\frac{w^{p-1}}{2}}(C^\circ)\cap S$.
\end{enumerate}
\end{dfn}

The last two conditions restricts the position of $\zeta$ in $C$ such that $\zeta$ lies sufficiently far from the boundary of $C$.
The site $m$ will be refereed to as a \emph{midpoint} of $C$.
Let $S$ be the support of the Pauli operator $E$,
any restriction of which obeys no-strings rule.
\begin{lemma}
Given a path $\zeta$ in $S$ joining $u_1$ and $u_n$ such that $d(u_1,u_n)=lw^p-1$,
there are $l$ disjoint maximal chunks of level-$p$ whose midpoints are on $\zeta$.
\label{lem:max_chunks_on_path}
\end{lemma}
\begin{lemma}
For sufficiently large $w$,
a maximal level-$(p+1)$ chunk $C$ with respect to $S$
admits a decomposition into $w+1$ or more maximal chunks of level-$p$ with respect to $S$.
\label{lem:many_max_in_max}
\end{lemma}

The proofs of Lemma~\ref{lem:max_chunks_on_path},~\ref{lem:many_max_in_max} can be found in Appendix~B.
The support of the minimal Pauli operator $E$ in Theorem~\ref{thm:separation_charged_defects}
must admit a path connecting $S'$ and $S \setminus S'$.
Otherwise, $S'$ can be regarded as being created locally,
and our topological order condition demands the cluster be neutral.
Since the path has length $\ge R$,
Lemma~\ref{lem:max_chunks_on_path} says we have a maximal chunk of level-$p$
where $p$ is such that $w^p \le R < w^{p+1}$.
Lemma~\ref{lem:many_max_in_max} implies any maximal chunk of level-$p$ must contain
at least $(w+1)^p$ sites.
This proves Theorem~\ref{thm:separation_charged_defects} with $\gamma=\frac{\log(w+1)}{\log w}>1$.
A similar argument proves the lower bound $d=\Omega(L^\gamma)$ on the code distance $d$ of Code 1 in \cite{Haah:2011}
since the minimal logical operator must contain a path of length $L$.

\section{Conclusions and open problems}

We proposed an argument in favor of self-correcting properties for
a class of 3D spin Hamiltonians with topological quantum order.
Thermal diffusion of defects in these Hamiltonians is suppressed
by the presence of logarithmic energy barriers.  A novel technique for
proving lower bounds on the energy barriers is presented.

Our results rise several questions. Firstly, Theorem~\ref{thm:main'}
implies that  the energy landscape contains a macroscopic number of
local minimums  separated
by macroscopic energy barriers. These minimums correspond to
low-lying excited states with only a few defects such that separation between
defects is of order $L$.  Such energy landscape suggests a possibility of
a spin glass phase at sufficiently low temperature. We note that a
spin glass phase can indeed be realized for some classical spin Hamiltonians
with logarithmic energy barriers such as the model discovered by Newman and Moore~\cite{NewmanMoore:1999}.
Interplay between topological order and spin glassiness has been studied recently by
several authors~\cite{Chamon:2005,Tsomokos:2010}.
Secondly, our paper leaves open the question of how to perform error correction necessary  to
extract encoded information from a memory. In particular, no efficient error correction
algorithm is known for the stabilizer code discovered in~\cite{Haah:2011}.
One could speculate however that locality and a macroscopic distance of the code
are sufficient conditions for having a non-zero error threshold under
random independent errors.  Finally, an exciting open question is whether
any 3D stabilizer code Hamiltonian with TQO has point-like topological defects.
A proof of this conjecture would give an evidence that  ``strong self-correction'' similar to the one
in the 4D toric code is impossible for  realistic physical systems.

\section*{Acknowledgments}
We would like to thank David DiVincenzo, John Preskill, Barbara Terhal, and Beni Yoshida for helpful discussions.
SB acknowledges hospitality of the  Institute for Quantum Information, Caltech, where part of this work was done.
SB was supported in part by the DARPA QuEST program under contract number HR0011-09-C-0047, and
JH in part by
NSF No. PHY-0803371,
ARO No. W911NF-09-1-0442,
DOE No. DE-FG03-92-ER40701,
and the Korea Foundation for Advanced Studies.

\section*{Appendix A}

\setcounter{lemma}{1}
In this section we prove Lemma~\ref{lemma:RG2}.
For convenience of the reader, we repeat the statement of the lemma.
\begin{lemma}
Let $S'\equiv S(t')$ and $S''\equiv S(t'')$ be a consecutive pair of syndromes in the level-$p$ syndrome history.
Let $E$ be the product of all errors $E_j$ that occurred between $S'$ and $S''$.
If $4m (2+\xi(p)) < L_{tqo}$,
then there exists an error $\tilde{E}$  supported on $\calB_{\xi(p)}(S'\cup S'')$
such that $E\tilde{E}$ is a stabilizer.
\end{lemma}
\begin{proof}
Let us  use induction in $p$. The base of induction is $p=0$.
In this case $E=E_j$ is a single-qubit error. If the qubit acted on by $E$ does
not belong to $\calB_{1}(S'\cup S'')$, one must have $S'=S''$. It means that 
$E$ is a single-qubit error with a trivial syndrome. Topological order condition implies that 
$E$ is a stabilizer. Choosing $\tilde{E}=I$ proves the
lemma for $p=0$.

Suppose we have already proved the statement for some level $p$
and let us prove for level $p+1$.
Let $S'=S(t')$ and $S''=S(t'')$ be consecutive syndromes in the level-$(p+1)$ 
history. Consider first the trivial case when 
$S'=S(t')$ and $S''=S(t'')$ are also consecutive syndromes in the level-$p$ history.
Then $S'$ and $S''$ are connected by a single level-$p$
error $E$ which, by induction hypothesis, has support on $\calB_{\xi(p)}(S'\cup S'')$
(modulo stabilizers).
The latter is contained in $\calB_{\xi(p+1)}(S'\cup S'')$ which proves the induction step.

The non-trivial case is when there is at least one level-$p$ syndrome between $S'$ and $S''$.
The interval of the level-$p$ syndrome history between $S'$ and $S''$ can be represented
(after properly redefining the time variable $t$) as
\[
S'\xrightarrow{E_{lead}}
S(1) \xrightarrow{E_1}
S(2) \xrightarrow{E_2}
\cdots
 \xrightarrow{E_{\tau-1}}
S(\tau) \xrightarrow{E_{tail}}
S''.
\]
Here all syndromes $S(1),\ldots,S(\tau)$ are sparse at the level $p$ and
all transitions are caused by single level-$p$ errors.
The condition of sparsity implies that 
the set of elementary cubes occupied by  $S(t)$ has a unique
partition into a disjoint union of clusters $C_a(t)$ such that
each cluster has diameter at most $\xi(p)$ and the distance
between any pair of clusters (if any) is at least
\bea
\mathrm{dist}(C_a(t),C_b(t)) & \ge &  \xi(p+1)-2\xi(p)\ge (10\alpha -2)\xi(p) \nn \\
 &\ge &  8\alpha \xi(p). \nn
\eea
Represent any intermediate syndrome as a disjoint union
\be
\label{S1+}
S(t)=S^c(t) \cup S^n(t), \quad t=1,\ldots,\tau,
\ee
where $S^c(t)$ and $S^n(t)$ include all charged and all neutral clusters $C_a(t)$ respectively.
Let $g$ be  the number of clusters in $S^c(t)$. We claim  that $g$  does not depend on $t$.
Indeed, since a level-$p$ error $E_t$ acts on $\xi(p)$-neighborhood of $S(t) \cup S(t+1)$,
the sparsity condition implies that $E_t$
cannot  create/annihilate a charged cluster $C_a(t)$ from the vacuum, or map a charged
cluster  to a neutral cluster (or vice verse). The same argument shows that
each  cluster $C_a(t)\subseteq  S^c(t)$ can `move' at most by $\xi(p)$ per time step,
that is, we can parameterize
\[
S^c(t)=C_1(t) \cup \ldots \cup C_g(t)
\]
such that
a `world-line' of the $a$-th charged cluster
obeys the continuity condition
\be
\label{cont+}
\mathrm{dist}(C_a(t+1),C_a(t)) \le \xi(p).
\ee
We can now use the no-strings rule to show that
all charged clusters are `locked' near their
initial positions, so that their world-lines
are essentially parallel to the time axis.  More precisely, we claim that
\be
\label{locking+}
\mathrm{dist}(C_a(t), C_a(1)) \le \alpha  \xi(p)  \quad \mbox{for all $1\le t \le \tau$}.
\ee
Indeed, suppose Eq.~\eqref{locking+} is false for some $a$.
Using the continuity Eq.~\eqref{cont+} one can find a time step $t_1$
such that $\mathrm{dist}(C_a(t_1), C_a(1))>\alpha \xi(p)$ and
$\mathrm{dist}(C_a(t), C_a(1)) \le \alpha \xi(p)$ for all $1\le t<t_1$.
Let $E_{close}$ be the product of all level-$p$ errors $E_j$
that occurred between $S(1)$ and $S(t_1)$ within distance $(2+\alpha)\xi(p)$
from $C_a(1)$. Since all intermediate syndromes are sparse at level $p$,
the net effect of $E_{close}$ is to annihilate the  charged cluster  $C_a(1)$
and create the charged cluster $C_a(t_1)$.
Equivalently, applying $E_{close}$ to the vacuum creates a pair of charged clusters
$C_a(1)$ and $C_a(t_1)$.
However, this contradicts to the no-strings rule since
$C_a(1)$ and $C_a(t_1)$ have linear size at most $\xi(p)$ while
the distance between them is greater than $\alpha \xi(p)$.
Thus we have proved Eq.~\eqref{locking+}.

We will say that $\vec{x}\in \Lambda$ is {\em close to $S'$} iff $\vec{x}\in \calB_{\xi(p+1)}(S')$.
We will say that $\vec{x}\in \Lambda$ is {\em close to $S''$} iff $\vec{x}\in \calB_{\xi(p+1)}(S'')$.

Let $E_t$ be the level-$p$ error causing the transition from  $S(t)$ to $S(t+1)$,
where $t=1,\ldots,\tau-1$. Let
$E^c_t$ is the restriction of $E_t$ onto $\calB_{\xi(p)}(S^c(t)\cup S^c(t+1))$,
and $E^n_t$ is the restriction of $E_t$ onto $\calB_{\xi(p)}(S^n(t)\cup S^n(t+1))$.
The sparsity of $S(t)$ and localization of level-$p$ errors then implies that
\be
\label{EcEn+}
E_t=E^c_t \cdot E^n_t.
\ee
We claim that {\em any error $E^c_t$ is close to $S'$}. Indeed,
each cluster  in $S^c(1)$ is within distance $2\xi(p)$ from $S'$ since
otherwise a single level-$p$ error $E_{lead}$ would be able
to create a charged cluster from the vacuum.  Using Eq.~\eqref{locking+}
we infer that $C_a(t) \in \calB_{(2+\alpha)\xi(p)}(S')$ for all $a=1,\ldots,g$.
Therefore $E^c_t$ is close to $S'$.

We shall now define a `localized' leading error $\tilde{E}_{lead}$
that maps the syndrome $S'$ to $S^c(1)$ such that the support of
$\tilde{E}_{lead}$ is close to $S'$.
For any neutral cluster $C \in S^n(1)$ let $O'(C)$ be a Pauli operator
creating $C$
from the vacuum. By definition of a neutral cluster,
we can assume that $O'(C)$ has support in $\calB_1(C)$.
Set
\[
\tilde{E}_{lead} = E_{lead} \cdot G \cdot \prod_{C \in S^n(1) } O'(C),
\]
where $G\in \calG$ is a stabilizer to be chosen later.
Since $E_{lead}$ is a single level-$p$ error, we have two options:
\begin{enumerate}
\item[(a)] any cluster $C\in S^n(1)$ is within distance $\xi(p)$ from $S'$,
\item[(b)] there is exactly one cluster  $C_{far}\in S^n(1)$ such that the distance
between $C_{far}$ and $S'$ is greater than $\xi(p)$.
\end{enumerate}
In case~(a) we set $G=I$ if the support of $E_{lead}$ is close to $S'$
and $G=E_{lead}$ otherwise (if $E_{lead}$ is not close to $S'$, but
any cluster $C\in S^n(1)$ is within distance $\xi(p)$ from $S'$,
the error $E_{lead}$ creates no defects when applied to the vacuum,
that is, $E_{lead}$ must be a stabilizer). In case~(b) the error $E_{lead}$
creates a neutral cluster $C_{far}$ from the vacuum. In this case
we will set $G=E_{lead} O'(C_{far})$. We conclude that the support
of $\tilde{E}_{lead}$ is close to $S'$ and $\tilde{E}_{lead}$ maps
$S'$ to $S^c(1)$.

We can apply the same rules to the error $E_{tail}$ and the syndrome $S^n(\tau)$.
It allows us to find a stabilizer $G$ and define a localized error
\[
\tilde{E}_{tail} = E_{tail} \cdot G \cdot \prod_{C \in S^n(\tau) } O''(C),
\]
such that $\tilde{E}_{tail}$ maps $S^c(\tau)$ to $S''$ and the support of $\tilde{E}_{tail}$
is close to $S''$. The operator $O''(C)$ above creates
a neutral cluster $C\in S^n(\tau)$ from the vacuum.

We can now define a localized level-$(p+1)$ error $\tilde{E}$
whose support is close to $S'\cup S''$  as
\be
\label{tildeE_cluster}
\tilde{E}=\tilde{E}_{lead} \cdot E^c_1 \cdots E^c_{\tau-1} \cdot \tilde{E}_{tail}.
\ee
By construction, it describes an error path
\[
S'\xrightarrow{\tilde{E}_{lead}}
S^c(1) \xrightarrow{E^c_1}
S^c(2) \xrightarrow{E^c_2}
\cdots
 \xrightarrow{E^c_{\tau-1}}
S^c(\tau) \xrightarrow{\tilde{E}_{tail}}
S''.
\]
It remains to check that $E\cdot \tilde{E}$ is a stabilizer.
Combining Eq.~\eqref{EcEn+} and
Eq.~\eqref{tildeE_cluster} we
conclude that
\[
E\cdot \tilde{E} = (\tilde{E}_{lead} E_{lead}) \cdot E^n_1 \cdots E^n_{\tau-1} \cdot (\tilde{E}_{tail} E_{tail}).
\]
Applying  $E\cdot \tilde{E}$ to the vacuum generates the following chain of transitions:
\bea
\label{chain1_cluster}
\mathrm{vac} \xrightarrow{\tilde{E}_{lead} E_{lead}}
S^n(1) \xrightarrow{E^n_1}
S^n(2) \xrightarrow{E^n_2} \cdots && \nn \\
\cdots
 \xrightarrow{E^n_{\tau-1}}
S^n(\tau) \xrightarrow{\tilde{E}_{tail} E_{tail}}
\mathrm{vac}. &&
\eea
Here all transitions are caused by errors whose support can be enclosed
by at most $m$ cubes of linear size $2+\xi(p)$.
Each syndrome $S^n(t)$
consists of at most $m$ neutral clusters of diameter $\xi(p)$, i.e., it can be
created from the vacuum by an error 
 whose support can be enclosed
by at most $m$ cubes of linear size $2+\xi(p)$.
Now the statement that $E\cdot \tilde{E}$ is a stabilizer follows from the following proposition.
\begin{prop}
\label{prop:1}
 Let $Q_j$ be Pauli operators 
causing  a chain of transitions
\[
\mathrm{vac} \xrightarrow{Q_1} S_1 \xrightarrow{Q_2}  S_2 \xrightarrow{Q_3} \ldots \xrightarrow{Q_r} S_r
\xrightarrow{Q_{r+1}} \mathrm{vac}.
\]
Let $P_j$ be some Pauli operator creating the syndrome $S_j$ from the vacuum.
Suppose the support of any operator $P_j$ and any operator $Q_j$
can be enclosed by at most $n$ cubes of linear size $R$ such that 
$4nR<L_{tqo}$. Then the product $\overline{Q}\equiv Q_1\cdots Q_r Q_{r+1}$ is a stabilizer.
\end{prop}
\begin{proof}[\bf Proof of Proposition~\ref{prop:1}]
Let $\psi_0$ be any ground state. Define a sequence of states
\bea
\psi(1) &=& P_1 Q_1\cdot \psi_0, \nn \\
\psi(j+1) &=&  (P_j P_{j+1}) Q_{j+1} \cdot \psi(j), \nn \\
\psi(r+1) &=& Q_{r+1} P_{r} \cdot \psi(r), \nn
\eea
for $j=0,\ldots,r-1$.
Obviously,
\[
\psi(j)=P_j \cdot (Q_1 \cdots Q_j)\cdot  \psi_0 \quad \mbox{for $j=1,\ldots,r$}
\]
and $\psi(r+1) =\overline{Q}\cdot \psi_0$.
It follows that all states $\psi(j)$ are ground states, and the transition from $\psi(j)$ to
$\psi(j+1)$ can be caused by a Pauli operator
\[
O_j=P_j P_{j+1} Q_{j+1}.
\]
Let $M_j$ be the support of $O_j$. Obviously, $M_j$ can be enclosed by at most $3n$ cubes of linear size $R$.
Suppose first that $M_j$ is a connected set,
i.e., one can connect any pair of qubits from $M_j$ by a path $(u_1,\ldots,u_l)$
such that the distance between $u_a$ and $u_{a+1}$ is $1$. 
Then $M_j$ can be enclosed by a single cube of linear size at most $3nR$.
Our assumptions on $n$ and $R$, and the topological order condition then implies that $O_j$ is a stabilizer.
Suppose now that $M_j$ consists of several disconnected components $M_j^\alpha$,
such that the distance between any pair of components is at least $2$.
Let $O_j^\alpha$ be the restriction of $O_j$ onto a connected component $M_j^\alpha$.
Locality of the stabilizer generators $G_a$ and the fact that $O_j$ commutes with any $G_a$
implies that $O_j^\alpha$ commutes with any $G_a$. Furthermore, $M_j^\alpha$
can be enclosed by a cube of linear size at most $3nR$. 
Topological order condition then implies that $O_j^\alpha$ is a stabilizer. 
It follows that all operator $O_j$ are stabilizers, that is, $\psi(j+1)=\psi(j)$ for all $j$.
However, it means  that $\overline{Q}\, \psi_0=\psi_0$ for any ground state $\psi_0$,
that is, $\overline{Q}$ is a stabilizer.
\end{proof}
\end{proof}

\section*{Appendix B}
In this section we prove the technical lemmas needed
for Theorem~\ref{thm:separation_charged_defects}.
\begin{lemma}
Given a path $\zeta$ in $S$ joining $u_1$ and $u_n$ such that $d(u_1,u_n)=lw^p-1$,
there are $l$ disjoint maximal chunks of level-$p$ whose midpoints are on $\zeta$.
\end{lemma}
\begin{proof}
For convenience, we assume that the $z$-coordinates of $u_1$ and $u_n$ are
$0$ and $lw^p-1$, respectively.
Consider $l+1$ planes $P_i$ perpendicular to the $z$-axis, whose $z$-coordinates are $iw^p$ for $i=0,1,\ldots,l$.
In each region between the two consecutive planes $P_{i-1}$ and $P_i$,
there is a subpath $\zeta_i=(u_{j_{i-1}},\ldots,u_{j_{i}})$ such that $d(u_{j_{i-1}},u_{j_{i}})=w^p-1$.
Choose $m_i \in \zeta_i$ such that $d(u_{j_{i-1}},m_i),d(m_i,u_{j_{i}}) \ge \frac{w^p -1}{2}$.
Let $C_i$ be the connected component of $m_i$ within $S \cap \calB_{\frac{w^p}{2}}(m_i)$.
Add, if necessary, some points of $\zeta_i$ to $C_i$ to get a maximally connected $C'_i$.
This $C'_i$ is a maximal chunk of sites with midpoint being $m_i$.
Any two $C'_i$'s are disjoint since each of them lies in a unique region enclosed by $P_{i-1}$ and $P_i$.
\end{proof}

\begin{lemma}
For sufficiently large $w$,
a maximal level-$(p+1)$ chunk $C$ with respect to $S$
admits a decomposition into $w+1$ or more maximal chunks of level-$p$ with respect to $S$.
\end{lemma}
\begin{proof}
Recall that $S$ is the support of the Pauli operator $E$, any restriction of which obeys no-strings rule.
Define the \emph{boundary} of a subset $U$ of $S$ to be $\partial U = \calB_1(U) \cap U^c \cap S$.
Then, any subset $U$ of sites with boundary enclosed in a two disjoint regions can be regarded as a string segment.

By the definition of the maximal chunk,
there exists a path $(u_1,\ldots,m,\ldots,u_n)$ in $C^\circ \subseteq C$ such that $d(u_1,u_n)=w^{p+1}-w^p$.
We assume that the $z$-coordinates of $u_1, u_n$ differ by $w^{p+1}-w^p$.
We will show that there are sufficiently long and separated paths in $C^\circ$,
to which we apply Lemma~\ref{lem:max_chunks_on_path} to find $w+1$ maximal chunks of level-$p$.
They will lie in $\calB_\frac{w^p}{2}(C^\circ)$, and hence in $C$.

Let $M$ ($N$) be the subset of $S$ consisted of sites
whose $z$-coordinates differ from that of $u_1$ ($u_n$) by at most $\eta w^p$.
First, suppose $\partial C^\circ$ is not contained in $M \cup N$.
Since $u_1 \in M$ and $u_n \in N$, there is a site $s \in C^\circ$
adjacent (of distance 1) to $\partial C^\circ$ such that $d(s,u_1),d(s,u_n) > \eta w^p$.
Furthermore, $d(s,m) \ge \frac{w^{p+1}-w^p}{2}-1$;
Otherwise, $C^\circ$ contains a site in the boundary, which is a contradiction.

Consider the shortest network $\mathcal{N}$ of paths in $C^\circ$ connecting four sites $u_1,m,u_n,s$.
(The length of a network of paths is the number of sites in the union of the paths.)
Let $\zeta$ be the shortest path in $\mathcal{N}$ from $u_1$ to $u_n$.
If $s$ is not contained in $\calB_{3w^p}(\zeta)$,
then $\zeta' \subseteq \mathcal{N}$ joining $s$ to a site on $\zeta$
has a subpath $\zeta'' \subseteq \zeta'$ of diameter at least $2w^p$
such that $\zeta''$ is separated from $\zeta$ by $w^p$.
Applying Lemma~\ref{lem:max_chunks_on_path} to $\zeta$ and $\zeta''$,
we find at least $w+1$ maximal chunk of level-$p$.
If $m$ is not contained in $\calB_{3w^p}(\zeta)$,
a similar argument reveals at least $w+1$ maximal chunk of level-$p$.

Suppose both $s$ and $m$ are contained in $\calB_{3w^p}(\zeta)$.
Observe that $\zeta \setminus ( \calB_{4w^p}(s) \cup \calB_{4w^p}(m) )$
consists of three connected components $\zeta_1,\zeta_2,\zeta_3$,
two of which have diameter $\ge \frac{w^{p+1}-w^p}{2}-8w^p$ and the other has diameter $\ge (\eta - 4)w^p$.
Two distinct $\calB_{\frac{w^p}{2}}(\zeta_i)$ and $\calB_{\frac{w^p}{2}}(\zeta_j)$ ($i,j=1,2,3$)
do not overlap because of the minimality of $\zeta$. Applying Lemma~\ref{lem:max_chunks_on_path},
we find $w + \eta - 21$ maximal chunks of level-$p$.
Choosing $\eta > 21$, we get the desired result.

Next, suppose $\partial C^\circ$ is contained in $M \cup N$.
Let $s_M, s_N \in C^\circ \setminus (M\cup N)$ be sites adjacent to $M$ and $N$, respectively.
The separation between $M$ and $N$ is $(w-1-2\eta)w^p$.
If it is greater than $\eta' \alpha w^p$, there must be a $z$-plane $P$ that contains $s_M$ or $s_N$
such that $P \cap C^\circ$ has diameter $> \eta' w^p$; Otherwise, the no-strings rule is violated.
Let $v_1, v_2 \in P \cap C^\circ$ be sites separated by $\eta' w^p$.
The four sites, $u_1,u_n,v_1,v_2$ are sufficiently separated and connected by some paths in $C^\circ$.
Arguing as before, we find $(w -1) + \eta'-16$ maximal chunks of level-$p$.
The choice of $\eta' > 17$ and $w > 1+2\eta+\eta'\alpha$ proves the lemma.
\end{proof}

\bibliography{mybib}
\end{document}